\documentclass[11pt,letterpaper,oneside,english]{article}
\usepackage{enumerate}

\usepackage{amssymb,amsbsy,latexsym}
\usepackage{amsmath,amsthm}

\usepackage[latin1]{inputenc}

\makeatletter
\let\@font@warningori\@font@warning
\newcommand\shutup{\def\@font@warning##1{}}
\newcommand\youcanspeak{\let\@font@warning\@font@warningori}
\makeatother 

\shutup
\usepackage{fourier}
\youcanspeak
\usepackage[scaled=0.875]{helvet}

\usepackage{amscd} 

\usepackage{pst-all}
\usepackage{pstricks-add}

\newtheorem{theorem}{Theorem}
\newtheorem{corollary}[theorem]{Corollary}

\newtheorem{remark}{Remark}
\newtheorem{definition}{Definition}

\providecommand{\setN}{\mathbb{N}}
\providecommand{\setZ}{\mathbb{Z}}
\providecommand{\setQ}{\mathbb{Q}}
\providecommand{\setR}{\mathbb{R}}

\newcommand{\lindisc}{\textrm{lindisc}}
\newcommand{\disc}{\textrm{disc}}
\newcommand{\herdisc}{\textrm{herdisc}}
\newcommand{\Dperm}{D^{\textrm{perm}}}
\newcommand{\Dmon}{D^{\textrm{mon}}}

\title{\Large Bin Packing via Discrepancy of Permutations\footnote{A preliminary version of this paper appeared in SODA'11 \cite{BinPackingViaPermutationsSODA2011}.}}
\date{}
\author{
Friedrich Eisenbrand\footnote{EPFL, Lausanne, Switzerland. Email: {\tt{friedrich.eisenbrand@epfl.ch}}. Supported by the Swiss National Science Foundation (SNF).} \\
\and
D{ö}m{ö}t{ö}r P{á}lv{ö}lgyi\footnote{Eötvös Loránd University (ELTE), Budapest, Hungary. Email: {\tt{dom@cs.elte.hu}}} \\
\and 
Thomas Rothvoß\footnote{M.I.T., Cambridge, USA. Email: {\tt{rothvoss@math.mit.edu}}. Supported by  the  German Research Foundation (DFG) within the Priority Program 1307 ``Algorithm Engineering'', by the Alexander von Humboldt Foundation within the Feodor Lynen program, by ONR grant N00014-11-1-0053 and by NSF contract CCF-0829878.} \\
}

\begin{document}

\maketitle

\begin{abstract} 
A well studied special case of \emph{ bin packing} is the \emph{3-partition  problem},
 where $n$ items of size $>\frac{1}{4}$ have to
be packed in a minimum number of bins of capacity one.  The famous
\emph{Karmarkar-Karp algorithm} transforms a fractional solution of a
suitable LP relaxation for this problem into an integral  solution that
requires at most $O(\log n)$ additional bins.

The \emph{three-permutations-problem} of Beck is the following.
Given any 3 permutations on $n$ symbols, color the symbols
 red and blue, such that in any interval of any of those
permutations, the number of red and blue symbols is roughly the same. 
The necessary difference is called the \emph{discrepancy}. 

We establish a surprising connection between bin packing and Beck's
problem: The additive integrality
gap of the { 3-partition} linear programming relaxation can be bounded
by the discrepancy of 3 permutations. 

This connection yields an alternative method to establish an $O(\log n)$ bound on the 
additive integrality gap of the {3-partition}. 
Reversely, making use of a recent example of 3 permutations, for which a discrepancy
of $\Omega(\log n)$ is necessary, we prove the following: 
The $O(\log^2 n)$ upper bound on the additive gap for bin packing with arbitrary item sizes cannot
be improved by any technique that is based on rounding up items. This lower bound holds
for a large class of algorithms including the Karmarkar-Karp procedure.
\end{abstract}




\section{Introduction}

The \emph{bin packing} problem is the following. Given 
 $n$ items of size $s_1,\ldots,s_n \in [0,1]$ respectively, the goal is 
to pack these items in as  few bins of capacity one as possible.
 Bin packing is a fundamental problem in
Computer Science with  numerous applications in theory and
practice. 


The development of heuristics for bin packing with better and better
performance guarantee is an important  success story in  the field of
\emph{Approximation Algorithms}. 
Johnson~\cite{Johnson73,JohnsonFFD74} has shown that the  \emph{First Fit}
algorithm  requires  at most $1.7\cdot OPT +
1$ bins and that  \emph{First Fit Decreasing} yields a solution with 
$\frac{11}{9} OPT + 4$
bins (see~\cite{FFDtightBound-Dosa07} for a tight bound of $\frac{11}{9} OPT + \frac{6}{9}$). 
An important step forward was made by  Fernandez de la Vega and
Luecker~\cite{deLaVegaLueker81} who provided an asymptotic polynomial
time approximation scheme for bin packing. The \emph{rounding
  technique} that is introduced in their paper has been very
influential in the design of PTAS's for many other difficult
combinatorial optimization problems. 


In 1982, Karmarkar and Karp~\cite{KarmarkarKarp82} proposed an
approximation algorithm for bin packing  that can be analyzed to yield
a solution using at most $OPT + O (\log^2 n)$ bins. This seminal
procedure is based 
on the  \emph{Gilmore Gomory LP
  relaxation}~\cite{Gilmore-Gomory61,TrimProblem-Eiseman1957}: 
\begin{equation} \label{eq:1}
  \begin{array}{lcl}
\min \sum_{p\in\mathcal{P}} x_p && \\
 \sum_{p\in\mathcal{P}} p\cdot x_p &\geq& \mathbf{1}  \\
 x_p &\geq& 0 \quad \forall p\in\mathcal{P} 
    
  \end{array}\tag{LP}
\end{equation}
Here $\mathbf{1} = (1,\ldots,1)^T$ denotes the all ones vector and
$\mathcal{P} = \{ p\in\{0,1\}^n \colon s^Tp \leq 1\}$ is the set of all
feasible \emph{patterns}, i.e.  every vector in $\mathcal{P}$ denotes
a feasible way to pack one bin.  Let $OPT$ and $OPT_f$ be the value of
the best integer and fractional solution  respectively.  The linear
program \eqref{eq:1}  has an
exponential number of
variables 
but still one can compute a basic solution $x$ with $\mathbf{1}^Tx \leq
OPT_f + \delta$ in time polynomial in $n$ and
$1/\delta$~\cite{KarmarkarKarp82} using the Gr{ö}tschel-Lovász-Schrijver
variant of the Ellipsoid method~\cite{GLS-algorithm-Journal81}. 

The procedure of Karmarkar and Karp \cite{KarmarkarKarp82} yields an
\emph{additive integrality gap} of $O(\log^2 n)$, i.e. $OPT \leq OPT_f +
O(\log^2 n)$, see also~\cite{Williamson98}.  This corresponds to an asymptotic FPTAS\footnote{An  asymptotic fully polynomial time approximation scheme (AFPTAS) is an  approximation algorithm that produces solutions of cost at most  $(1+\varepsilon)OPT + p(1/\varepsilon)$ in time polynomial in $n$ and $1/\varepsilon$, where  also $p$ must be a polynomial.} for { bin packing}.
The authors in \cite{BinPacking-MIRUP-ScheithauerTerno97}  conjecture
that even $OPT \leq \lceil OPT_f\rceil + 1$ holds and this even if one replaces
the right-hand-side $\mathbf{1}$ by any other positive integral vector
$b$.  
This \emph{Modified Integer Round-up Conjecture} was proven by
Seb{\H{o}} and  Shmonin~\cite{MIRUPproofForDim7-SeboShmonin09}
if the number of different item sizes is at most $7$.  We would like
to mention that Jansen and
Solis-Oba~\cite{DBLP:conf/ipco/JansenS10} recently provided an $OPT+1$
approximation-algorithm for  bin packing if the number of item
sizes is fixed. 

Much of the hardness of { bin packing} seems to appear already in the
special case of \emph{ 3-partition}, where $3n$ items of size
$\frac{1}{4} < s_i < \frac{1}{2}$ with $\sum_{i=1}^{3n} s_i = n$ have
to be packed.  It is strongly $\mathbf{NP}$-hard to distinguish
between $OPT \leq n$ and $OPT \geq n+1$~\cite{GareyJohnson79}.  No
stronger hardness result is known for general { bin packing}.  A
closer look into~\cite{KarmarkarKarp82} reveals that, with the
restriction $s_i > \frac{1}{4}$, the Karmarkar-Karp algorithm uses
$OPT_f + O(\log n)$ bins\footnote{The \emph{geometric grouping}
  procedure (Lemma~5 in \cite{KarmarkarKarp82}) discards items of size
  $O(\log \frac{1}{s_{\min}})$, where $s_{\min}$ denotes the size of
  the smallest item. The geometric grouping is applied $O(\log n)$
  times in the Karmarkar-Karp algorithm.  The claim follows by using
  that $s_{\min} > \frac{1}{4}$ for { 3-partition}.}.

\subsection*{Discrepancy theory}

Let $[n] := \{ 1,\ldots,n\}$ and consider a set system  $\mathcal{S} \subseteq
2^{[n]}$ over the ground set $[n]$. 
A \emph{coloring} is a mapping $\chi : [n] \to \{ ± 1\}$.  In \emph{discrepancy
theory}, one aims 
at finding colorings  for which the difference of ``red'' and ``blue''
elements in all sets is as small as possible. Formally, the
\emph{discrepancy} of a set system $\mathcal{S}$ is defined as
\[
 \disc(\mathcal{S}) = \min_{\chi : [n] \to \{ ± 1\}} \max_{S\in \mathcal{S}} |\chi(S)|.
\]
where  $\chi(S) = \sum_{i\in S} \chi(i)$. A random coloring
provides an easy bound of $\disc(\mathcal{S}) \leq O(\sqrt{n \log
  |\mathcal{S}|})$ \cite{GeometricDiscrepancy-Matousek99}.  The famous
``\emph{Six Standard Deviations suffice}'' result
of Spencer~\cite{SixStandardDeviationsSuffice-Spencer1985} improves this to
$\disc(\mathcal{S}) \leq O(\sqrt{n
  \log(2|\mathcal{S}|/n)})$. 

If every element appears in at
most $t$ sets, then the \emph{Beck-Fiala
Theorem}~\cite{IntegerMakingTheorems-BeckFiala81} yields
$\disc(\mathcal{S}) < 2t$. The same authors conjecture that in fact
$\disc(\mathcal{S}) = O(\sqrt{t})$.
Srinivasan~\cite{DiscrepancyBound-sqrtT-logN-SrinivasanSODA97} gave a
$O(\sqrt{t} \log n)$ bound, which was improved by
Banaszczyk~\cite{BalancingVectors-Banaszczyk98} to $O(\sqrt{t \log
  n})$.  Many such discrepancy proofs are purely existential, for
instance due to the use of the pigeonhole principle. In a very recent
breakthrough Bansal~\cite{DiscrepancyMinimization-Bansal2010} showed
how to obtain the desired colorings for the
Spencer~\cite{SixStandardDeviationsSuffice-Spencer1985} and
Srinivasan~\cite{DiscrepancyBound-sqrtT-logN-SrinivasanSODA97} bounds
by considering a random walk, guided by the solution of a semidefinite
program.  

For several decades, the following
\emph{three-permutations-conjecture} or simply \emph{Beck's
  conjecture} (see Problem 1.9 in \cite{DiscrepancyTheory-BeckSos95}) was open: 
\begin{quote}
  Given any 3 permutations on $n$ symbols, one can color the symbols
  with red and blue, such that   in every interval of every of those
  permutations, the number of red and   blue symbols differs by $O(1)$. 
\end{quote}
Formally, a set of permutations  $\pi_1,\ldots,\pi_k : [n] \to [n]$  induces a
set-system\footnote{We only consider intervals of permutations that start
 from the first element. Since any interval is the difference of two
 such prefixes, this changes the discrepancy by a factor
of at most $2$. } 
 $$\mathcal{S} =  \{ \{ \pi_i(1),\ldots,\pi_i(j)\}  \colon  j=1,\ldots,n; \; i=1,\ldots,k\}.$$ 
We denote the maximum discrepancy 
of such a set-system induced by   $k$ permutations over $n$ symbols as
$\Dperm_k(n)$, then Beck's conjecture can be rephrased as $\Dperm_3(n) = O(1)$. %
One can provably upper bound $\Dperm_3(n)$ by $O(\log n)$ and more
generally $\Dperm_k(n)$ can be bounded by~$O(k\log
n)$~\cite{DiscrepancyOf3Permutations-Bohus90} and by $O(\sqrt{k}\log
n)$~\cite{DiscrepancyBound-sqrtT-logN-SrinivasanSODA97,DiscrepancyOfPermutations-SpencerEtAl}
using the so-called \emph{entropy method}.

But very recently a counterexample to Beck's conjecture was found by Newman and Nikolov~\cite{CounterexampleToBecksConjecture-NewmanNikolov2011}
(earning a prize of 100 USD offered by Joel Spencer)\footnote{The counterexample was announced few months after SODA'11. As a small anecdote, both authors of \cite{CounterexampleToBecksConjecture-NewmanNikolov2011} had a 
joint paper~\cite{HardnessForDiscrepancySODA2011} on a related topic, 
which was presented in the same session of SODA'11 as the conference version of this paper.}.
In fact, they fully settle the question by proving that $\Dperm_3(n) = \Theta(\log n)$.

\subsection*{Our contribution}

The first result of this paper is the following theorem.
\begin{theorem}
  \label{thr:1}
  The additive integrality gap of the linear
  program \eqref{eq:1} restricted to 3-partition instances is
  bounded by $6\cdot \Dperm_3(n)$. 
\end{theorem}
This result is constructive in the following sense. If 
one can find a $\alpha$ discrepancy coloring
for any three permutations in polynomial time, then there is an $OPT + O(\alpha)$ approximation
algorithm for 3-partition. 

The proof of Theorem~\ref{thr:1} itself  is  via  two steps.

\begin{enumerate}[i)]
\item We show that the additive
  integrality gap of~\eqref{eq:1} is at most twice the maximum
  \emph{linear discrepancy} of a \emph{$k$-monotone matrix} if all
  item sizes are larger than  $1/(k+1)$ 
  (Section~\ref{sec:from-bin-packing}). 
  This step is based on matching techniques and \emph{Hall's
    theorem}. 
   \label{item:1}  
\item We then show that the linear discrepancy of a $k$-monotone
  matrix is at most $k$ times the discrepancy of $k$
  permutations (Section~\ref{sec:BoundingMatDiscByPermDisc}). This result uses
  a theorem of Lovász, Spencer and   Vesztergombi. \label{item:2} 
\end{enumerate}
The theorem then follows by setting $k$ equal to $3$ in the above
steps. 

Furthermore, we show  that the discrepancy of $k$ permutations is at most $4$
times the linear discrepancy of a $k$-monotone matrix.  Moreover in Section~\ref{sec:kLogN-bound-for-k-mon-matrices}, we
provide a $5k\cdot \log_2(2\min\{m,n\})$ upper bound on  the linear
discrepancy of a $k$-monotone $n×m$-matrix.

Recall that most approximation algorithms for bin packing or corresponding 
generalizations rely on ``rounding up items'', i.e. they select some patterns
from the support of a fractional solution which form a valid solution to a 
dominating instance. Reversing the above connection, we can show that no
algorithm that is only based on this principle can obtain an additive 
integrality of $o(\log n)$ for item sizes $>\frac{1}{4}$ and $o(\log^2 n)$
for arbitrary item sizes (see Section~\ref{sec:ImplicationsFromCounterexample}).
This still holds if we allow to discard and greedily pack items. More precisely:

\begin{theorem}
For infinitely many $n$, there is a bin packing instance $s_1 \geq \ldots \geq s_n > 0$ with 
a feasible fractional (LP) solution $y \in [0,1]^{\mathcal{P}}$ such that the following holds: 
Let $x \in \setZ_{\geq 0}^{\mathcal{P}}$ be an integral solution and $D \subseteq [n]$ be those items that are not covered by $x$ 
with the properties:
\begin{itemize}
\item \emph{Use only patterns from fractional solution:} $\textrm{supp}(x) \subseteq supp(y)$.
\item \emph{Feasibility:} $\exists \sigma : [n] \setminus D \to [n]$ with $\sigma(i) \leq i$ and $\sum_{p:i \in p} x_p \geq |\sigma^{-1}(i)|$ for all $i \in \{ 1,\ldots,n\}$.
\end{itemize}
Then one has $\mathbf{1}^T x + 2\sum_{i \in D} s_i \geq \mathbf{1}^T y + \Omega(\log^2 n)$.
\end{theorem}
Improving the Karmarkar-Karp algorithm has been a longstanding open
problem for many decades now. Our result shows that the recursive
rounding procedure of the algorithm is optimal.
In order to break the $O(\log^2 n)$ barrier it does not suffice to
consider only the patterns that are contained in an initial fractional
solution as it is the case for the Karmarkar-Karp algorithm.




\section{Preliminaries}
\label{sec:furth-prel-discr}





We first  review some  further necessary  preliminaries on discrepancy
theory. We refer to~\cite{GeometricDiscrepancy-Matousek99} for further details. 

If $A$ is a matrix, then we denote the $i$th row of $A$ by $A_i$ and
the $j$th entry in the $i$th row by $A_{ij}$.  The notation of
discrepancy can be naturally extended to real matrices $A\in\setR^{m×n}$
as
\[
  \disc(A) := \min_{x\in\{ 0,1\}^{n}} \|A(x - 1/2 \cdot \mathbf{1})\|_{\infty},
\]
see, e.g.~\cite{GeometricDiscrepancy-Matousek99}.
Note that if $A$ is the incidence matrix of a set system $\mathcal{S}$
(i.e. each row of $A$ corresponds to the characteristic vector of a
set $S \in \mathcal{S}$), then $\disc(A) =
\frac{1}{2}\disc(\mathcal{S})$, hence this notation is consistent ---
apart from the $\frac{1}{2}$ factor.

The \emph{linear discrepancy} of a matrix $A\in\setR^{m×n}$  is defined as
\[
  \lindisc(A) := \max_{y\in[0,1]^n} \min_{x\in\{ 0,1\}^n} \|Ax - Ay\|_\infty.
\]
This value can be also described by a two player game. The first
player chooses a fractional vector $y$, then the second player chooses
a $0/1$ vector $x$. The goal of the first player is to maximize, of
the second to minimize $\|Ax - Ay\|_\infty$.
The inequality  $\disc(A) \leq \lindisc(A)$ holds by choosing $y := (1/2,\ldots,1/2)$.
One more notion of defining the ``complexity'' of a set system or a matrix is that of
the \emph{hereditary discrepancy}:
\[
 \herdisc(A) := \max_{B \textrm{ submatrix of } A} \disc(B).
\]
Notice that one can assume that $B$ is formed by choosing
a subset of the columns of $A$.  
This parameter is obviously at least $\disc(A)$ since we can choose
$B:=A$ and in \cite{DiscrepancyofSetSystemsAndMatrices-LSV86} even an
upper bound for $\lindisc(A)$ is proved (see again
\cite{GeometricDiscrepancy-Matousek99} for a recent description).

\begin{theorem}[(Lovász, Spencer, Vesztergombi)] \label{thm:herdisc}
For $A \in \setR^{m×n}$ one has  $$\lindisc(A) \leq 2\cdot\herdisc(A).$$
\end{theorem}


\section{Bounding the gap via the discrepancy  of monotone matrices} 
\label{sec:from-bin-packing}

A matrix $A$ is called $k$\emph{-monotone} if all its column vectors have non-decreasing
entries from $0,\ldots,k$. In other words $A \in \{0,\ldots,k\}^{m×n}$ and $A_{1j} \leq \ldots \leq A_{mj}$ for any column $j$.
We denote the maximum linear discrepancy of such matrices by
\[
  \Dmon_k(n) := \max_{\substack{A\in\setZ^{m×n} \\ k\textrm{-monotone}}} \lindisc(A).
\]
The next theorem establishes step~\ref{item:1}) mentioned in the
introduction. 

\begin{theorem} \label{thm:BinPackingGap}
Consider the linear program \eqref{eq:1} and suppose that the item
sizes satisfy  $s_1,\ldots,s_n > \frac{1}{k+1}$.  
Then 
\[
  OPT \leq OPT_f + \left(1 + \frac{1}{k}\right) \Dmon_k(n).
\]
\end{theorem}
\begin{proof}
Assume that the item sizes are sorted such that  $s_1 \geq \ldots \geq s_n$.
Let $y$ be any optimum basic solution of \eqref{eq:1} and
let $p_1,\ldots,p_m$ be the list of patterns. 
Since $y$ is a basic solution, its support satisfies
$|\{ i \colon  y_i > 0\}| \leq n$. Hence by deleting unused patterns, we
may assume\footnote{In case that there are less than $n$ patterns, we add empty patterns.} that $m=n$. 

We define  $B = (p_1,\ldots,p_n) \in \{0,1\}^{n×n}$ 
as the matrix composed of the patterns as column vectors. Clearly  $By = \mathbf{1}$. 
%
Let $A$ be the matrix that is defined by $A_i := \sum_{j=1}^i B_j$, again $A_i$ denotes the $i$th row of $A$. 
In other words, $A_{ij}$ denotes the number of items of types $1,\ldots,i$
in pattern $p_j$. Since $By = \mathbf{1}$  we have $Ay =
(1,2,3,\ldots,n)^T$. 
Each column of $A$ is monotone. Furthermore, since no pattern contains
more than $k$ items one has  $A_{ij} \in \{ 0,\ldots,k\}$, thus $A$ is
$k$-monotone. 

We attach a row $A_{n+1} := (k,\ldots,k)$ as the new last row of $A$. 
Clearly $A$ remains $k$-monotone. 
There exists a vector $x\in\{0,1\}^n$ with 
\[
\|Ax - Ay\|_{\infty} \leq \lindisc(A) \leq \Dmon_k(n).
\]
We \emph{buy} $x_i$ times pattern $p_i$ and $\Dmon_k(n)$ times the pattern that
only contains the largest item of size $s_1$.

It remains to show: (1) this yields a feasible solution; (2)
the number of patterns does not exceed the claimed bound of $OPT_f + (1+\frac{1}{k})\cdot\Dmon_k(n)$.

For the latter claim, recall that the constraint emerging from row $A_{n+1}=(k,\ldots,k)$ 
together with $\sum_{i=1}^n y_i = OPT_f$ provides
\[
 k\sum_{i=1}^n x_i \leq 
k\cdot\sum_{i=1}^n y_i + \Dmon_k(n) = k\cdot OPT_f+\Dmon_k(n).
\]
We use this to upper bound the number of opened bins by
\[
\sum_{i=1}^n x_i + \Dmon_k(n) \leq OPT_f + \Big(1+\frac{1}{k}\Big)\cdot\Dmon_k(n). 
\] 
It remains to prove that our integral solution is feasible.
To be more precise, we need to show that any item $i$ can be assigned to a space
reserved for an item of size $s_i$ or larger.

\begin{figure}[ht] 
\begin{center}
\begin{pspicture}(0,0)(3,5)
\psellipse[linecolor=gray,linewidth=1.5pt, linestyle=dashed](-0.1,2.75)(0.75,1.8)\psellipse[linecolor=gray,linewidth=1.5pt, linestyle=dashed](2.1,2.75)(0.75,1.8) \rput(3.4,2.25){\gray{$N(V')$}}

\cnode*(0,4){2pt}{v1} \nput[labelsep=0pt]{180}{v1}{$v_1$}
\cnode*(0,3){2pt}{v2} \nput[labelsep=0pt]{180}{v2}{$v_2$}
\cnode*(0,1.5){2pt}{v3} \nput[labelsep=0pt]{180}{v3}{$v_i$}
\cnode*(0,0){2pt}{v4} \nput[labelsep=0pt]{180}{v4}{$v_n$}
\rput[c](0,0.75){$\vdots$}
\rput[c](0,2.25){$\vdots$}
 \rput(-1.1,2.75){\gray{$V'$}}
\cnode*(2,4){2pt}{u1} \nput[labelsep=1pt]{0}{u1}{$u_1$}  \nput[labelsep=0.8]{0}{u1}{$b_1=B_1x + \Dmon_k(n)$}
\cnode*(2,3){2pt}{u2} \nput[labelsep=1pt]{0}{u2}{$u_2$} \nput[labelsep=0.8]{0}{u2}{$b_2=B_2x$}
\cnode*(2,1.5){2pt}{u3} \nput[labelsep=1pt]{0}{u3}{$u_i$} \nput[labelsep=0.8]{0}{u3}{$b_i=B_ix$}
\cnode*(2,0){2pt}{u4} \nput[labelsep=1pt]{0}{u4}{$u_n$} \nput[labelsep=0.8]{0}{u4}{$b_n=B_nx$}
\rput[c](2,0.75){$\vdots$}
\rput[c](2,2.25){$\vdots$}
\ncline{v1}{u1}
\ncline{v2}{u1}
\ncline{v3}{u1}
\ncline{v4}{u1}
\ncline{v2}{u2}
\ncline{v3}{u2}
\ncline{v4}{u2}
\ncline{v3}{u3}
\ncline{v4}{u3}
\ncline{v4}{u4}

\nput[labelsep=18pt]{90}{v1}{$V$}
\nput[labelsep=18pt]{90}{u1}{$U$}
\end{pspicture}
\end{center}

 \caption{The bipartite graph in the proof of Theorem~\ref{thm:BinPackingGap}}\label{fig:1}
  
\end{figure}
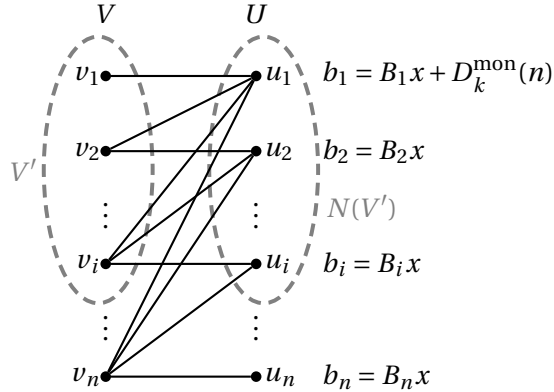

To this end, consider  a bipartite graph with nodes $V=\{v_1,\ldots,v_n\}$ on the left,
representing the items. The  nodes on the right are the
set 
$U = \{ u_1,\ldots,u_n\}$, where each $u_i$ is attributed
with a multiplicity $b_i$ representing  the number of times that
we reserve space for items of size $s_i$ in our solution, see
Figure~\ref{fig:1}. Recall that 
\[
  b_i = \begin{cases} B_ix + \Dmon_k(n) & \textrm{if } i=1 \\ 
B_ix & \textrm{otherwise} \end{cases}. 
\]
We insert an edge $(v_i,u_j)$ for all $i\geq j$. 
The meaning of this edge is the following. One 
can assign item $i$ into the space which is reserved
for item $j$ since $s_i \leq s_j$. We claim that there exists a $V$-perfect
matching, respecting the multiplicities of $U$. 
By \emph{Hall's Theorem}, see, e.g.~\cite{Diestel97}, it suffices to
show for any subset $V' \subseteq V$ 
that the multiplicities of the nodes in $N(V')$ (the neighborhood of
$V'$) are at least $|V'|$. Observe that 
$N(v_{i}) \subseteq N(v_{i+1})$, hence it suffices to prove the claim for sets of
the form $V' = \{ 1,\ldots,i\}$. 
For such a $V'$ one has 
\begin{eqnarray*}
\sum_{u_j\in N(V')} b_j = \Dmon_k(n) + \sum_{j=1}^i B_jx  
= \Dmon_k(n) + A_ix \geq A_iy  = i 
\end{eqnarray*}
and the claim follows.
\end{proof}

\section{Bounding the discrepancy of monotone matrices by the discrepancy of  permutations} 
\label{sec:BoundingMatDiscByPermDisc}

In this section, we show that the linear discrepancy of $k$-monotone
matrices is essentially bounded by the discrepancy of $k$
permutations. 
This corresponds to step~\ref{item:2}) in the proof of the main
theorem.   
By Theorem~\ref{thm:herdisc} it suffices to bound the discrepancy of
$k$-monotone matrices by the discrepancy of $k$ permutations times a
suitable factor. 

We first  explain how one can
associate a permutation  to a $1$-monotone matrix. Suppose that $B \in \{0,1\}^{m×n}$ is
a $1$-monotone matrix. 
If $B^j$ denotes the $j$-th column of
$B$, then the permutation $\pi$ that we associate with $B$ is the
(not necessarily unique) permutation  that satisfies $B^{\pi(1)}\geq B^{\pi(2)}\geq \cdots \geq
B^{\pi(n)}$  where $u\geq v$ for vectors $u,v\in \setR^m$ if $u_i\geq v_i$ for
all $1\leq i\leq m$. On the other hand the matrix $B$ (potentially
plus some extra rows and after merging identical rows) gives the
incidence matrix of the set-system induced by $\pi$. 



A $k$-monotone matrix  $B$
can be decomposed into a sum of $1$-monotone matrices $B^1,\ldots,B^k$. 
Then any $B^{\ell}$ naturally corresponds
to a permutation $\pi_{\ell}$ of the columns as we explained above. 
A low-discrepancy coloring of these permutations yields a coloring that
has low discrepancy for any $B^{\ell}$ and hence also for $B$, as we
show now in detail. 

\begin{theorem} \label{thm:kMonDiscAtMostPermDisc} 
For any $k,n\in\setN$, one has $\Dmon_k(n) \leq k\cdot \Dperm_k(n)$.
\end{theorem}

\begin{proof}
  Consider any $k$-monotone matrix $A\in\setZ^{m×n}$.  
  By virtue of Theorem \ref{thm:herdisc}, there is a
  $m× n'$ submatrix, $B$, of $A$ such that $\lindisc(A)\leq 2\cdot \disc(B)$, thus
  it suffices to show that $\disc(B) \leq \frac{k}{2}\cdot\Dperm_k(n)$. Of
  course, $B$ itself is again $k$-monotone.

Let $B^{\ell}$ also be a $m× n'$ matrix, defined by
\[
B^{\ell}_{ij} := \begin{cases} 1 & \textrm{if } B_{ij} \geq \ell \\ 0, & \textrm{otherwise.} \end{cases}
\]
The matrices $B^{\ell}$ are $1$-monotone, and the matrix $B$ decomposes into  $B = B^1 + \ldots + B^k$. 
As mentioned above, for any $\ell$, there is a (not necessarily unique) 
permutation $\pi_{\ell}$ on $[n']$
such that $B^{\ell,\pi_\ell(1)} \geq B^{\ell,\pi_\ell(2)} \geq \ldots \geq B^{\ell,\pi_\ell(n')}$, where $B^{\ell,j}$
denotes the $j$th column of $B^{\ell}$. Observe that the row vector $B_{i}^{\ell}$
is the characteristic vector of the set $\{\pi_{\ell}(1),\ldots,\pi_{\ell}(j)\}$, 
where $j$ denotes the number of ones in $B_i^{\ell}$.

Let $\chi:[n']\to\{±1\}$ be the coloring that has discrepancy at most $\Dperm_k(n)$
with respect to all permutations $\pi_1,\ldots,\pi_k$. In particular $|B_i^{\ell}\chi| \leq \Dperm_k(n)$,
when interpreting $\chi$ as a $± 1$ vector. 
Then by the triangle inequality
\[
\disc(B) \leq \frac{1}{2}\|B\chi\|_{\infty} \leq \frac{1}{2}\sum_{\ell=1}^{k} \|B^{\ell}\chi\|_{\infty} 
\leq \frac{k}{2}\Dperm_k(n).
\]
\end{proof}

Combining Theorem~\ref{thm:BinPackingGap} and Theorem~\ref{thm:kMonDiscAtMostPermDisc}, 
we conclude
\begin{corollary}
Given any { bin packing} instance with $n$ items of size bigger than $\frac{1}{k+1}$
one has
\[
  OPT \leq OPT_f + 2k\cdot\Dperm_k(n).
\]
\end{corollary}

In particular, this proves Theorem~\ref{thr:1}, our main result.

\subsection*{Bounding the discrepancy of permutations in terms of the discrepancy of monotone matrices}
\label{sec:BoundingPermDiscByMatDisc}

In addition we would like to note that the discrepancy of permutations
can be also bounded by the discrepancy of $k$-monotone matrices as
follows. 

\begin{theorem}\label{thm:PermutationDiscAtMostkMonDisc} 
For any $k,n\in\setN$, 
one has $\Dperm_k(n) \leq 4\cdot\Dmon_k(n).$
\end{theorem}
\begin{proof}
We will show that for any permutations $\pi_1,\ldots,\pi_k$ on $[n]$, there is a $kn× n$ $k$-monotone matrix $C$ with $\disc(\pi_1,\ldots,\pi_k) \leq 4\cdot\disc(C)$.
Let $\Sigma \in \{ 1,\ldots,n\}^{kn}$ be the string which we obtain
by concatenating the $k$ permutations. That means
$\Sigma = (\pi_1(1),\ldots,\pi_1(n),\ldots,\pi_k(1),\ldots,\pi_k(n))$. Let $C$ the matrix where $C_{ij}$ is the number of appearances of $j\in\{1,\ldots,n\}$ among the
first $i\in\{1,\ldots,kn\}$ entries of $\Sigma$. By definition, $C$ is $k$-monotone, in fact it is the ``same'' $k$-monotone matrix as in the previous proof. 

Choose $y := (\frac 12,\dots, \frac 12)$ to have $Cy=(\frac{1}{2},1,\ldots,\frac{kn}{2})$.
Let $x\in\{0,1\}^n$ be a vector with $\|Cx- Cy\|_{\infty} \leq \disc(C)$. 
Consider the coloring $\chi: [n] \to \{ ± 1\}$ with $\chi(j) := 1$ if $x_j = 1$
and $\chi(j) := -1$ if $x_j = 0$. We claim that the discrepancy of this coloring is bounded by $4\cdot\disc(C)$ for all $k$ permutations.
Consider any prefix $ S:= \{ \pi_i(1),\ldots,\pi_i(\ell)\}$. 
Let $r = C_{(i-1)n + \ell} \in \{ i-1,i\}^n$ be the row of $C$ that corresponds to this
prefix. 
With these notations we have
\begin{eqnarray*}
  |\chi(S)| \leq 
|(r - (i-1)\mathbf{1})\cdot(2x-2y)| 
\leq 2\cdot\big( \underbrace{|r(x-y)|}_{\leq\disc(C)} + \underbrace{|k\cdot\mathbf{1}(x-y)|}_{\leq\disc(C)}\big) 
\leq 4\cdot\disc(C).
\end{eqnarray*}
Here the inequality $|(k\cdot\mathbf{1})\cdot(x-y)| \leq \disc(C)$ comes from
the fact that $k\cdot\mathbf{1} = (k,\ldots,k)$ is the last row of $C$.
\end{proof}





\section{A bound on the discrepancy of monotone matrices} 
\label{sec:kLogN-bound-for-k-mon-matrices}

Finally, we want to provide a non-trivial upper bound on the linear
discrepancy of $k$-monotone matrices. 
The result of Spencer, Srinivasan and
Tetali~\cite{DiscrepancyOfPermutations-SpencerEtAl,DiscrepancyBound-sqrtT-logN-SrinivasanSODA97}
together with Theorem~\ref{thm:kMonDiscAtMostPermDisc} yields a bound
of $\Dmon_k(n) = O(k^{3/2} \log n)$.  This bound can be reduced by a
direct proof that shares some similarities with that of
Bohus~\cite{DiscrepancyOf3Permutations-Bohus90}.  Note that
$\Dmon_k(n) \geq k/2$, as the $k$-monotone $1×1$ matrix $A = (k)$
together with target vector $y=(1/2)$ witnesses. 

\begin{theorem} \label{thm:LogBoundForLinearDisc}
Consider any $k$-monotone matrix $A\in\setZ^{n×m}$. Then
\[
  \lindisc(A) \leq 5k\cdot \log_2 (2\min\{ n,m\}).
\]
\end{theorem}
\begin{proof}
If $n=m=1$, $\lindisc(A) \leq \frac{k}{2}$, hence the claim is true. Let $y \in [0,1]^m$ by any vector. 
We can remove all columns $i$ with $y_i = 0$ or $y_i=1$ and then apply induction (on the size of the matrix).
Next, if $m > n$, i.e. the number of columns is bigger then the number of constraints,
then $y$ is not a basic solution of the system
\begin{eqnarray*}
Ay &=& b \\
0 \leq y_i &\leq& 1 \quad \forall i=1,\ldots,m.
\end{eqnarray*}
We replace $y$ by a basic solution $y'$ and apply induction (since $y'$ has some integer
entries and $Ay = Ay'$). 

Finally it remains to consider the case $m \leq n$. Let $a_1,\ldots,a_n$ be the
rows of $A$ and let $d(j) := \|a_{j+1} - a_j\|_1$ for $j=1,\ldots,n-1$, i.e. $d(j)$
gives the cumulated differences between the $j$th and the $(j+1)$th row. Since the columns are $k$-monotone,
each column contributes at most $k$ to the sum $\sum_{j=1}^{n-1} d(j)$. Thus
\[
  \sum_{j=1}^{n-1} d(j) \leq mk \leq nk.
\]
By the pigeonhole principle at least $n/2$ many rows $j$ have $d(j) \leq 2k$.
Take any second of these rows and we obtain a set $J \subseteq \{ 1,\ldots,n-1\}$ of size $|J| \geq n/4$
such that for every $j\in J$ one has $d(j) \leq 2k$ and $(j+1)\notin J$.
Let $A'y = b'$ be the subsystem of $n' \leq \frac{3}{4}n$ many equations, which we
obtain by deleting the rows in $J$ from $Ay = b$. We apply induction to this system and
obtain an $x\in\{0,1\}^{m}$ with 
\begin{eqnarray*}
\|A'x - A'y\|_{\infty} &\leq& 5k\cdot \log_2(2n') \\
&\leq& 5k \log_2\Big(2\cdot\frac{3}{4}n\Big) \\
&\leq& 5k\log_2(2n) - 5k\log_2\Big(\frac{4}{3}\Big) \\
&\leq& 5k \log_2(2n) - 2k.
\end{eqnarray*}
Now consider any $j\in\{1,\ldots,n\}$.
If $j\notin J$, then row $j$ still appeared in $A'y = b'$, hence $|a_j^Tx - a_j^Ty| \leq 5k\log_2(2n) - 2k$.
Now suppose $j\in J$. We remember that $j+1\notin J$, thus $|a_{j+1}^T(x-y)| \leq 5k\log_2(2n) - 2k$. But then
using the triangle inequality
\begin{eqnarray*}
  |a_{j}^Tx - a_j^Ty| &\leq& \underbrace{|(a_{j+1} - a_j)^T(x-y)|}_{\leq d(j) \leq 2k} + \underbrace{|a_{j+1}^T(x-y)|}_{\leq 5k\log_2(2n) - 2k} 
\leq 5k\cdot \log(2n).
\end{eqnarray*}
\end{proof}

\section{Lower bounds for algorithms based on rounding up items}
\label{sec:ImplicationsFromCounterexample}

Let us remind ourselves, how the classical approximation algorithms for bin packing work. 
For example in the algorithm of de la Vega and Lueker \cite{deLaVegaLueker81} one first \emph{groups} the items, i.e. 
the item sizes $s_i$ are rounded up to some $s_i' \geq s_i$ such that (1) the number of different item sizes in $s'$ 
is at most $O(1/\varepsilon^2)$ (for some proper choice of $\varepsilon$) and (2) the optimum number of bins increases only by a $(1+\varepsilon)$ factor.
Note that any solution for the new instance with bigger item sizes induces a solution with the same value for the
original instance. Then one computes a basic solution\footnote{Alternatively one can compute an optimum solution for the rounded instance by dynamic programming in time  $n^{(1/\varepsilon)^{O(1/\varepsilon)}}$, but using the LP reduces the running time to $f(\varepsilon)\cdot n$.}  $y \in \setQ_{\geq 0}^{\mathcal{P'}}$ 
to (LP) with $|\textrm{supp}(y)| \leq O(1/\varepsilon^2)$ and uses $(\lceil y_p\rceil)_{p \in \mathcal{P}'}$ as approximate solution (here $\mathcal{P}'$ are the feasible patterns
induced by sizes $s'$). 

In contrast, the algorithm of Karmarkar and Karp~\cite{KarmarkarKarp82} uses an iterative procedure, where
in each of the $O(\log n)$ iterations, the
item sizes are suitably rounded and the integral parts $\lfloor y_p\rfloor$ from a basic solution $y$ are bought. 
Nevertheless, 
both algorithms rely only on the following properties of bin-packing:
\begin{itemize}
\item \emph{Replacement property:} If $p$ is a feasible pattern (i.e. $\sum_{i \in p} s_i \leq 1$) with $j \in p$ and $s_i \leq s_j$, then 
$(p \backslash \{ j\}) \cup \{ i\} $ is also feasible. 
\item \emph{Discarding items:} Any subset $D \subseteq [n]$ of items can be greedily assigned
to at most $2s(D) + 1$ many bins ($s(D) := \sum_{i \in D} s_i$).
\end{itemize}

For a vector $x \in \setZ_{\geq 0}^{\mathcal{P}}$, we say that $x$ \emph{buys $\sum_{p \in \mathcal{P}: i \in p} x_p$ many slots for item $i$}. 
The replacement property implies that e.g. for two items $s_1 \geq s_2$; $x$ induces a feasible solution 
already if it buys no slot for item $2$, but 2 slots for the larger item $1$.

In the following we always assume that $s_1 \geq \ldots \geq s_n$. We say that an integral vector $x$ \emph{covers} the non-discarded items $[n] \setminus D$, 
if there is a map $\sigma : [n]\backslash D \to [n]$ with $\sigma(i) \leq i$ and $\sum_{p \in \mathcal{P}: i \in p} x_p \geq |\sigma^{-1}(i)|$.
Here the map $\sigma$ assigns items $i$ to a slot that $x$ reserves for an item of size $s_{\sigma(i)} \geq s_i$.
In other words, a tuple $(x,D)$ corresponds to a feasible solution if $x$ covers the items in $[n] \setminus D$
and the cost of this solution can be bounded by $\mathbf{1}^Tx + 2s(D)+1$.

It is not difficult to see\footnote{Proof sketch: Assign input items $i$ iteratively in
increasing order (starting with the largest one) to the smallest available slot. 
If there is none left for item $i$, then there are less then $i$ slots for items $1,\ldots,i$.} that for the existence of such a mapping $\sigma$ 
it is necessary (though i.g. not sufficient) that
\begin{equation}  \label{eq:CoveringInequality}
  \sum_{p \in \mathcal{P}} x_p \cdot |p \cap \{ 1,\ldots,i\}|   \geq i - |D| \quad \forall i\in[n].
\end{equation}
The algorithm of Karmarkar and Karp starts from a fractional solution $y$
and obtains a pair $(x,D)$ with $\mathbf{1}^Tx \leq \mathbf{1}^Ty$ and $\sum_{i \in D} s_i = O(\log^2 n)$ such
that $x$ covers $[n] \setminus D$. Moreover, it has the property\footnote{The Karmarkar-Karp method solves the (LP) $O(\log n)$ many times for smaller and smaller instances. This can either be done by reoptimizing the previous fractional solution or by starting from scratch. We assume here that the first option is chosen.} that $\textrm{supp}(x) \subseteq \textrm{supp}(y)$, which means that
it only uses patterns that are already contained in the 
support of the fractional solution $y$. Hence this method falls into an abstract class
of algorithms that can be characterized as follows:

\begin{definition}  \label{def:BasedOnRoundingUpItems}
We call an approximation algorithm for bin packing
\emph{based on rounding up items}, if for given item sizes $s_1,\ldots,s_n$ and a given fractional 
solution $y \in [0,1]^{\mathcal{P}}$ to (LP) it performs as follows: 
The algorithm produces a tuple $(x,D)$ such that 
(1) $x \in \setZ_{\geq 0}^{\mathcal{P}}$, (2)  $\textrm{supp}(x) \subseteq \textrm{supp}(y)$ and (3) $x$ covers $[n] \setminus D$.
We define the \emph{additive integrality gap} for a tuple $(x,D)$ as
\[
  \mathbf{1}^Tx + 2\sum_{i \in D} s_i - \mathbf{1}^Ty.
\]
\end{definition}
We can now argue that the method of Karmarkar and Karp is optimal
for all algorithms that are based on rounding up items. 
The crucial ingredient is the recent result of Newman and Nikolov~\cite{CounterexampleToBecksConjecture-NewmanNikolov2011} that there 
are 3 permutations of discrepancy $\Omega(\log n)$. 
For a permutation $\pi$ we let $\pi([i]) = \{ \pi(1),\ldots,\pi(i)\}$ be the prefix 
consisting of the first $i$
symbols. In the following, let $\mathbb{O} = \{ \ldots,-5,-3,-1,1,3,5,\ldots\}$ be the set of
odd integers.

\begin{theorem}{\cite{CounterexampleToBecksConjecture-NewmanNikolov2011}} \label{thm:NewmanNikolov}  
For every $k \in \setN$ and $n = 3^k$, there are permutations $\pi_1,\pi_2,\pi_3 : [n] \to [n]$ such that
$\disc(\pi_1,\ldots,\pi_3) \geq k/3$. Additionally, for every coloring $\chi : [n] \to \mathbb{O}$ one has:
\begin{itemize}
\item If $\chi([n]) \geq 1$, then there are $i,j$ such that $\chi(\pi_j([i])) \geq (k+2)/3$ 
\item If $\chi([n]) \leq -1$, then there are $i,j$ such that $\chi(\pi_j([i])) \leq -(k+2)/3$.
\end{itemize}
\end{theorem}
Note that the result of \cite{CounterexampleToBecksConjecture-NewmanNikolov2011} was 
only stated for $\{± 1\}$ colorings. But the proof uses only the fact that
the colors $\chi(i)$ are odd integers\footnote{The only point where \cite{CounterexampleToBecksConjecture-NewmanNikolov2011} uses that $\chi(i) \in \{ ± 1\}$ 
is the base case $k=1$ of the induction in the proof of Lemma~2. In fact, the 
case $\chi([3]) \geq 1$ with a single positive symbol $i \in \{ 1,2,3\}$ becomes possible if
one considers colorings with odd numbers. However, also this case can easily be seen to be true. Interestingly, coloring all multiples of 3 with $+2$ and all other numbers with $-1$
would yield a constant discrepancy.}.
This theorem does not just yield a $\Omega(\log n)$ discrepancy, but also the stronger claim that 
any  coloring $\chi$ which is balanced (i.e. $|\chi([n])|$ is small) yields a prefix of
one of the permutations which has a ``surplus'' of $\Omega(\log n)$ and 
another prefix that has a ``deficit'' of $\Omega(\log n)$. 

We begin with slightly reformulating the result. Here we make no attempt to optimize any constant.
A \emph{string} $\Sigma=(\Sigma(1),\ldots,\Sigma(q))$ is an ordered sequence; $\Sigma(\ell)$ denotes the symbol 
at the $\ell$th position and $\Sigma[\ell] = (\Sigma(1),\ldots,\Sigma(\ell))$ denotes the prefix string consisting of 
the first $\ell$ symbols. 
We write $\chi(\Sigma[\ell]) = \sum_{i=1}^{\ell} \chi(\Sigma(i))$ and $\mathbb{O}_{\geq -1} = \{ -1,1,3,5,\ldots\}$.
\begin{corollary}  \label{cor:LogNDiscString}
For infinitely many even $n$, there is a string $\Sigma \in [n]^{3n}$, each of the $n$ symbols appearing
exactly $3$ times, such that: 
for all $\chi : [n] \to \mathbb{O}_{\geq -1}$ with $\chi([n]) \leq \frac{\log n}{40}$, there is an even $\ell \in \{ 1,\ldots,3n\}$
with $\chi(\Sigma[\ell]) \leq - \frac{\log n}{20}$.
\end{corollary} 

Note that this statement is in fact true for every large enough $n$ using a similar argument but we omit the proof as for us this weaker version suffices.

\begin{proof}
For some $k \in \setN$, let $\pi_1,\pi_2,\pi_3$ be the permutations on $[3^k]$ according to Theorem~\ref{thm:NewmanNikolov}.
We append the permutations together to a string $\Sigma$ of length $3\cdot3^k$. 
Additionally, for $n := 3^k+1$, we append 3 times the symbol $n$ to $\Sigma$. Thus
\[
  \Sigma = (\pi_1(1),\ldots,\pi_{1}(3^k),\pi_2(1),\ldots,\pi_2(3^k),\pi_3(1),\ldots,\pi_3(3^k),n,n,n)
\]
and $\Sigma$ has even length.

Next, let $\chi : [n] \to \mathbb{O}_{\geq-1}$ be any coloring with $|\chi([n])| \leq \frac{\log n}{40}$.
Reducing the values of at most $\frac{1}{2}(\frac{\log n}{40}+1)$ colors by 2, 
we obtain a coloring $\chi' : [n] \to \mathbb{O}_{\geq -1}$ with $\chi'([3^k])\leq-1$.
Then by Theorem~\ref{thm:NewmanNikolov} there are $j\in\{1,\ldots,3\}$ and $i \in \{ 1,\ldots,3^k\}$ 
such that $\chi'(\pi_j([i])) \leq -(k+2)/3$. For $\ell := (j-1)\cdot 3^k + i$ one has
\[
\chi(\Sigma[\ell]) \leq \chi'(\Sigma[\ell])+3(\frac{\log n}{40}+2) \leq (j-1)\cdot\underbrace{\chi'([3^k])}_{<0} + \underbrace{\chi'(\pi_j[i])}_{\leq-(k+2)/3} + 3(\frac{\log n}{20} + 2) \leq -\frac{\log n}{20}
\]
for $n$ large enough. If $\ell$ is not even, we can increment it by $1$ --- the discrepancy is
changed by at most $2$ (since we may assume that the last symbol $\Sigma(\ell)$ is negative, thus $-1$),
which can be absorbed into the slack that we still have.
\end{proof}

\subsection{A $\Omega(\log n)$ lower bound for the case of item sizes $>1/4$\label{sec:LogNLowerBound}}

In the following, for an even $n$, let $\Sigma$ be the string from 
Cor.~\ref{cor:LogNDiscString}.
We define a matrix $A \in \{ 0,1\}^{3n × n}$ such that
\[
  A_{ij} := \begin{cases} 1 & \Sigma(i) = j \\ 0 & \textrm{otherwise}. \end{cases}
\]
Note that $A$ has a single one entry per row and 3 one entries per column.

Next, we add up pairs of consecutive rows to obtain a matrix $B \in \{ 0,1,2\}^{(3/2)n × n}$.
Formally $B_i := A_{2i-1} + A_{2i}$. 
We define a bin packing instance by choosing item sizes $s_i := \frac{1}{3} - \varepsilon i$ for items $i=1,\ldots,\frac{3}{2}n$ with $\varepsilon := \frac{1}{20n}$.
Then $\frac{1}{3} > s_1 > s_2 > \ldots > s_{(3/2)n} > \frac{1}{4}$.
Furthermore we consider $B$ as our pattern matrix and $y := (\frac{1}{2},\ldots,\frac{1}{2})$ 
a corresponding feasible fractional solution. Note that $By = \mathbf{1}$.

In the following theorem we will assume for the sake of contradiction that this instance
admits a solution $(x,D)$ respecting Def.~\ref{def:BasedOnRoundingUpItems} with additive gap $o(\log n)$. It
is not difficult to see, that then $|D| = o(\log n)$ and $|\mathbf{1}^Tx - \mathbf{1}^Ty| = o(\log n)$.
The integral vector $x$ defines a coloring $\chi : [n] \to \mathbb{O}_{\geq -1}$ via the equation $x_i = y_i + \frac{1}{2}\chi(i)$.
This coloring is balanced, i.e. $|\chi([n])| = o(\log n)$. Thus there is a prefix string $\Sigma[\ell]$ 
with a deficit of $\chi(\Sigma[\ell]) \leq -\Omega(\log n)$. 
This corresponds to $x$ having $\ell/2 - \Omega(\log n)$ slots for the 
largest $\ell/2$ items, which implies that $x$ cannot be feasible. Now the proof in detail: 

\begin{theorem} \label{thm:logNLowerBoundForLargeItems}
There is no algorithm for bin packing,
based on rounding up items which achieves an additive integrality gap of $o(\log n)$
for all instances with $s_1,\ldots,s_n > 1/4$.
\end{theorem}
\begin{proof}
Let $(x,D)$ be a solution to the constructed instance with $\textrm{supp}(x) \subseteq \textrm{supp}(y)$ such that $x$
is integral and covers the non-discarded items $[\frac{3}{2}n] \setminus D$. 
For the sake of contradiction assume that
\[
  \mathbf{1}^Tx + 2\sum_{i \in D} s_i \leq \mathbf{1}^Ty + o(\log n).
\]
Clearly we may assume that $\mathbf{1}^Tx \leq \mathbf{1}^Ty + \frac{1}{600} \log n$, otherwise there is nothing to show.
Note that $\mathbf{1}^Tx + 2s(D) \geq \frac{(3/2)n-|D|}{3} + 2\cdot\frac{|D|}{4} = \mathbf{1}^Ty + \frac{|D|}{6}$ (since $\frac{1}{3}>s_i>\frac{1}{4}$) 
and thus $|D| \leq \frac{1}{100}\log n$.
Furthermore $\mathbf{1}^Tx \geq \frac{(3/2)n-|D|}{3} \geq  \mathbf{1}^Ty - \frac{1}{300} \log n$. We can summarize:
\begin{equation} \label{eq:ConditionForGoodX}
 |\mathbf{1}^Tx-\mathbf{1}^Ty| \leq \frac{\log n}{300} \quad \textrm{ and } \quad \sum_{i'=1}^i B_{i'}x \geq i - \frac{\log n}{100}  \quad \forall i \in [\frac{3}{2}n]
\end{equation}
We will now lead this to a contradiction.
Recall that every symbol $i \in \{ 1,\ldots,n\}$ corresponds to a column of matrix $B$. 
Define a coloring $\chi : [n] \to \mathbb{O}_{\geq -1}$ such that $x_i = \frac{1}{2} + \frac{1}{2}\chi(i)$.  
Note that indeed the integrality of $x_i$ implies that $\chi(i)$ is an odd integer.
Furthermore $|\chi([n])| = 2\cdot |\mathbf{1}^Tx - \mathbf{1}^Ty| \leq \frac{1}{150} \log n$. Using Cor.~\ref{cor:LogNDiscString}
there is a $2q \in \{ 1,\ldots,3n\}$ such that $\chi(\Sigma[2q]) \leq -\frac{\log n}{20}$.
The crucial observation is that by construction $\chi(\Sigma[2q]) = \sum_{i=1}^{q} B_i\chi$.
Then the number of slots that $x$ reserves for the largest $q$ items is
\[
  \sum_{i=1}^{q} B_{i}x = \underbrace{\sum_{i=1}^{q} B_iy}_{=q} + \frac{1}{2}\sum_{i=1}^{q} B_i\chi = q + \frac{1}{2} \underbrace{\chi(\Sigma[2q])}_{\leq -\frac{\log n}{20} } \leq q - \frac{\log n}{40}.
\]
Thus $x$ cannot cover items $[\frac{3}{2}n] \setminus D$.
\end{proof}

\subsection{A $\Omega(\log^2 n)$ lower bound for the general case}

%

Starting from the pattern matrix $B$ defined above, we will construct 
another pattern matrix $C$ and a vector $b$ of item multiplicities 
such that for the emerging instance even a $o(\log^2 n)$ additive integrality
gap is not achievable by just rounding up items. 

Let $\ell := \log n$ be a parameter. We will define 
\emph{groups} of items for every $j=1,\ldots,\ell$, where group $j \in \{ 1,\ldots,\ell\}$
contains $\frac{3}{2}n$ many different item types; each one 
with multiplicity $2^{j-1}$.
Define 
\[
  C := \begin{pmatrix}
       2^0\cdot B & \mathbf{0} & \mathbf{0} & \ldots & \mathbf{0} \\
       \mathbf{0} & 2^1\cdot B & \mathbf{0} & \ldots & \mathbf{0} \\
       \mathbf{0} & \mathbf{0} & 2^2 \cdot B & \ldots & \mathbf{0} \\
    \vdots & \vdots & \vdots & \ddots & \vdots \\
    \mathbf{0} & \mathbf{0} & \mathbf{0} & \ldots & 2^{\ell-1}\cdot B 
\end{pmatrix} \quad \textrm{and} \quad
b = \begin{pmatrix} 2^0\cdot\mathbf{1} \\
2^1\cdot\mathbf{1} \\
2^2\cdot\mathbf{1} \\
\vdots  \\
2^{\ell-1}\cdot\mathbf{1} \\
\end{pmatrix},
\]
thus $C$ is an $\frac{3}{2}n\ell × n\ell$ matrix and $b$ is a $\frac{3}{2}n\ell$-dimensional vector.
In other words, each group is a scaled clone of the instance in the previous section.
Choosing again $y := (1/2,\ldots,1/2) \in \setR^{\ell n}$ as fractional solution, 
we have $Cy = b$. Note that allowing multiplicities is just for 
notational convenience and does not make the problem
setting more general.  
Since the total number of items is still bounded by 
a polynomial in $n$ (more precisely $\mathbf{1}^Tb \leq O(n^2)$), each item $i$ could still be replaced by $b_i$
items of multiplicity $1$. 
Let $s_i^j := \frac{1}{3}\cdot(\frac{1}{2})^{j-1} - i\cdot\varepsilon$ the size of the $i$th item in group $j$
for  $\varepsilon := \frac{1}{12n^3}$.
Note that the size contribution of each item type is $2^{j-1}\cdot s_i^j \in [\frac{1}{3}, \frac{1}{3}-\frac{1}{n}]$.
Abbreviate the number of different item types by $m := \ell\cdot \frac{3}{2}n$. 

\begin{theorem} \label{thm:log2NLowerBound}
There is no algorithm for bin packing which is based on rounding up items
and achieves an additive integrality gap of $o(\log^2 n)$.
\end{theorem}
\begin{proof}
Let $(x,D)$ be arbitrary with $\textrm{supp}(x) \subseteq \textrm{supp}(y)$ such that $x$
is integral and covers $[m] \setminus D$ (considering $D$ now as a multiset). 
Assume for the sake of contradiction that 
\[
  \mathbf{1}^Tx + 2\sum_{i \in D} s_i \geq \mathbf{1}^Ty + o(\log^2 n).
\]
As in Theorem~\ref{thm:logNLowerBoundForLargeItems}, we can assume 
that $|\mathbf{1}^Tx - \mathbf{1}^Ty| \leq \frac{1}{10000} \log^2 n$.
First, observe that the bins in $y$ are packed pretty tight, i.e. $|\mathbf{1}^Ty - s([m])| \leq 1$.
If an item $i$ is covered by a slot for a larger item $i'$, then this
causes a \emph{waste} of $s_{i'} - s_i$, which is not anymore available for any other item.
The additive gap is defined as
$\mathbf{1}^Tx + 2s(D) - \mathbf{1}^Ty \geq s([m]/D) + \textrm{waste} + 2s(D) - s([m])-1 =  \textrm{waste} + s(D) - 1$.
Thus both, the waste and the size of the discarded items $s(D)$ must be bounded by $\frac{1}{10000}\log^2 n$.

Observe that the items in group $j-1$ are at least a factor $3/2$ larger than the 
items in group $j$.  
In other words, every item $i$ from group $j$ which is mapped to group $1,\ldots,j-1$
generates a waste of at least $\frac{1}{2} s_i^j$. Thus the total size of items which are
mapped to the slot of an item in a larger group is bounded by $\frac{1}{5000} \log^2 n$.
Thus there lies no harm in discarding these items as well --- let $D'$ be the union
of such items and $D$.
Then $s(D') \leq s(D) + \frac{1}{5000}\log^2 n \leq \frac{1}{3000}\log^2 n$.

For group $j$, let $D_{j}' \subseteq D'$ be the discarded items in the $j$th group and
let $x^j$ ($y^j$, resp.) be the vector $x$ ($y$, resp.), restricted to 
the patterns corresponding to group $j$. In other words, $x=(x^1,\ldots,x^{\ell})$ and $y=(y^1,\ldots,y^{\ell})$.
By $x_i^j \in \setZ_{\geq 0}$ we denote the
entry belonging to column $(\mathbf{0},\ldots,\mathbf{0},2^{j-1} B^i,\mathbf{0},\ldots,\mathbf{0})$ in $C$.
Pick $j \in \{ 1,\ldots,\ell\}$ uniformly 
at random, then $E[|\mathbf{1}^Tx^j - \mathbf{1}^Ty^j|] \leq \frac{1}{10000} \log n$
and $E[s(D_j')] \leq \frac{1}{3000} \log n$. By Markov's inequality, there must
be an index $j$, such that $|\mathbf{1}^Tx^j - \mathbf{1}^Ty^j| \leq \frac{1}{1000}\log n$
and $s(D_j') \leq \frac{\log n}{2000}$. Recall that $|D_j'| \leq 4\cdot2^{j-1} s(D_j') \leq 2^{j-1}\frac{\log n}{500}$. Since $x^j$ covers all items in group $j$ (without $D_j'$), we obtain
\[
  \sum_{i'=1}^i 2^{j-1} B_{i'}x^j \geq i\cdot2^{j-1} - 2^{j-1}\frac{\log n}{500} \quad \forall i=1,\ldots,\frac{3}{2}n
\]
After division by $2^{j-1}$, this implies Condition \eqref{eq:ConditionForGoodX}, 
which leads to a contradiction.

The claim follows since the number of items counted with multiplicity is bounded by $O(n^2)$, thus $\log^2 (\mathbf{1}^Tb) = \Theta(\log^2 n)$.
\end{proof}
%
%
%



\begin{remark}
Note that the additive integrality gap for the constructed instance is still small, once arbitrary patterns
may be used. For example a First Fit Decreasing assignment will produce a solution
of cost exactly $OPT_f$. This can be partly fixed by slightly increasing the item sizes. 
For the sake of simplicity consider the
construction in Section~\ref{sec:LogNLowerBound} and observe that the used patterns 
are still feasible if the items corresponding to the first permutation have sizes
in the range $[\frac{1}{3}+10\delta,\frac{1}{3} + 11\delta]$ and the items corresponding to the 2nd
and 3rd permutation have item sizes in $[\frac{1}{3}-7\delta,\frac{1}{3}-6\delta]$ (for a small constant $\delta >0$).
Then a First Fit Decreasing approach will produce a $\Omega(n)$ additive gap.
\end{remark}


\bibliographystyle{plain}
\bibliography{discrepancy}
\end{document}